\documentclass[epjST]{svjour}
\usepackage{graphics}
\usepackage{amssymb}
\usepackage{calrsfs}   
\usepackage{amsmath}
\usepackage[linesnumbered,ruled]{algorithm2e}
\usepackage{float}

\mathchardef\ordinarycolon\mathcode`\:
\mathcode`\:=\string"8000
\begingroup \catcode`\:=\active
  \gdef:{\mathrel{\mathop\ordinarycolon}}
\endgroup

\newcommand{\private}{\mathrm{private}}
\newcommand{\public}{\mathrm{public}}
\newcommand{\ego}{\mathrm{ego}}
\newcommand{\node}{\mathrm{node}}
\newcommand{\nodes}{\mathrm{nodes}}
\newcommand{\merger}{\mathrm{merger}}

\begin{document}

\title{Dominating Sets and Ego--Centered Decompositions in Social Networks} 

\author{Moses A. Boudourides \and Sergios T. Lenis}

\institute{University of Patras, Greece}

%\date{Received: date / Revised version: date}
\date{Revised version: date}

\abstract{
Our aim here is to address the problem of decomposing a whole network into a minimal number of ego--centered subnetworks. For this purpose, the network egos are picked out as the members of a minimum dominating set of the network. However, to find such an efficient dominating ego--centered construction, we need to be able to detect all the minimum dominating sets  and to compare all the corresponding dominating ego--centered decompositions of the network. To find all the minimum dominating sets of the network, we are developing a computational heuristic, which is based on the partition of the set of nodes of a graph into three subsets, the always dominant vertices, the possible dominant vertices and the never dominant vertices, when the domination number of the network is known. To compare the ensuing dominating ego--centered decompositions of the network, we are introducing a number of structural measures that count the number of nodes and links inside and across the ego--centered subnetworks. Furthermore, we are applying the techniques of graph domination and ego--centered decomposition for six empirical social networks.
}

\maketitle

\section{Introduction}
\label{intro}

In principle, there are two kinds of data collection designs for social network studies: either ``whole--network'' or ``egocentric'' designs \cite{RefMars}. Whole--network studies examine sets of interrelated actors linked by sets of relationships, the way they are extracted from certain relevant structural data sets \cite{RefWF}. The  whole--network data collection designs gather, without exception, all available structural information about the concerned sets of actors and their relationships, as far as the examined social network object may be regarded as already constituting a distinctive and bounded social collectivity (with the important caveat, as Peter Marsden has apprised, that ``network boundaries are often permeable and/or ambiguous'' \cite{RefWF}).

On the other side, the egocentric data collection designs prioritize only  particular focal actors and a limited range of their interactions by sampling and filtering out  the pertinent information from a larger population of emprical structural data. Thus, an {\bf{ego--centered network}} (or {\bf{ego--net}} or {\bf{personal network}}) \cite{RefCBe} hinges on some actor in the network, called the {\bf{ego}}, and the set of actors surrounding the ego, called {\bf{alters}}. Usually, the network data sets for ego--centered networks originate from General Social Surveys \cite{RefSMH} or other measurements conducted in the context of social science surveys or/and questionnaires \cite{Marsden}.

Although whole--network and egocentric designs are conceptually and operationally distinct, they can be considered to be interrelated from a methodological point of view. At one end, given an extensive whole network, there is a trivial (but quite redundant) assemblage of a large number of ego--centered networks spanning the whole network, when each actor in the latter would have been treated as a separate ego  \cite{RefMarsE}. On the other end, an egocentric design, in which egos are sampled as much densely as possible, may be reassembled to give an emerging whole--network construction \cite{RefKirke} (although as the outcome of a rather supernumerary data collection design).

Here, we are concentrating on the methodology of detecting an ego--centric decomposition of a whole network by following a rather parsimonious formal procedure. We are asking the following question: Given a whole network, how can one select a possibly minimal number of actors in such a way that, when these actors are considered as distinct egos, the total collection of the ego--centered networks, which are formed in this way, might reach all other actors and span the entire whole network? Of course, the answer to this question depends on what exactly one means by ``reaching''--``spanning'' and it is our intention to formulate this question in formal graph--theoretic terms by conceptualizing on the all--embracing structural patterns of adjacency, in which actors are embroiled in the network. In this way, a possible answer might be given by resorting to the formal graph--theoretic concept of the ``minimum dominating set,'' which is defined as the minimal set of actors whose members are adjacent to all other actors in the network. In other words, we propose a structurally conclusive solution to the problem of the efficient construction of a set of ego--centered networks, into which a given whole network might reduce, by introducing a decomposition of the whole network into smaller ego--centered subnetworks, which are built around the actors of a minimum dominating set of the whole network.\footnote{Of course, in such a decomposition, a minimum dominating set literally ``spans'' the vertex set. However, what is equally important to know is how the graph edge set is distributed  among the constituent ego--centered subnetworks, a question that we intend to tackle at the end of section 3 (in Corollary 2).} Apparently, the problem at hand is an optimization problem in the sense that it focuses on finding the minimum dominating set of actors that achieves a fixed goal, i.e., in our case, the goal of reaching (``being adjacent'') to all other actors in the network. Note that, in the literature of social network analysis, one often encounters the reverse problem as well: to find a set of actors of fixed size that achieves a certain structural goal (as, for instance, the goal of optimally diffusing something through the network \cite{RefBor}). 

Of course, it is well known that computationally the dominating set problem is $\mathcal{NP}$--complete \cite{RefJ}, although there are certain efficient algorithms (typically based on integer linear programming) for finding approximate solutions. Here, we are using the algorithms implemented in the Python--based SageMath software in order to find just the domination number of a given graph (network) and, subsequently, we are following our own methodology in order to detect all possible minimum dominating sets in the graph (network). The key point of our approach is based on a computational heuristics in partitioning the set of all actors (vertices) of a network (graph) into three sets, the always dominant, the possible dominant and the never dominant actors (vertices), when the domination number of the network is known. It turns out that knowing this partition simplifies the efficient computation of all minimum dominating sets (of course, always under the limitations imposed by the complexity of the problem on the size of the network). After defining the dominating ego--centered decomposition of a network in this way, we proceed in comparing the ensuing dominating ego--centered subnetworks through a number of structural measures that count the number of vertices and edges (relationships) inside and across the applied  dominating decomposition.

So, in the second section of our  investigation, we are discussing the fundamentals of the graph domination theory for social networks and sketch out our methodology for detecting all the minimum dominating sets of a given network (graph). In the third section, we are outlining the formal construction of the dominating ego--centered decomposition of a social network. Moreover, we are introducing the relevant structural measures and indices that arise in the context of this decomposition. Finally, in the last section, we are applying our techniques for the computation of all minimum dominating sets and the construction of the corresponding dominating ego--centered network decompositions to six examples of well known empirical social networks.

\section{Dominating Sets in Graphs}
\label{doms}

Let $G = (V, E)$ be a simple undirected graph with set of vertices $V$, where $|V| = n$,\footnote{For any set $X$, we denote by $|X|$ the {\bf{cardinality}} of $X$, i.e., the number of elements of $X$.} and set of edges $E = \{(u, v)\!: u \sim v, {\mbox{ for some }} u, v \in V, u \neq v\}$, where $u \sim v$ denotes that vertices $u$ and $v$ are adjacent (with the understanding that adjacency is a symmetric relationship, i.e., that edges $(u, v)$ and $(v, u)$ are identical). The open neighborhood of vertex $u$ is denoted as $N(u) = \{v \in V\!: (u, v) \in E\} = \{v \in V\!: v \sim u\}$, the closed neighborhood of $u$ is denoted as $N[u] = N(u) \cup \{u\}$, the degree of $u$ as $\deg(u) = |N(u)|$ and the (geodesic) distance between $u$ and $v$ as $d(u, v)$. For any $A \subset V$, we denote $N(A) = \bigcup_{v \in A} N(v), N[A] = N(A) \cup A$ and $d(u, A) = \min \{d(u, v)\!: v \in A\}$. Moreover, we assume that $G$ is the underlying graph of a social network and that the vertices of $G$ represent the actors of the social network, while the edges of $G$ represent the relationships/ties among actors in the social network. Thus, from now on, we are going to use interchangeably the terms network--graph, vertices--actors and edges--relationships/ties. 

\begin{definition}
\cite{RefHHS} Let $D$ be a set of vertices of graph $G$ ($D \subset V$).
\begin{itemize}
\item[\small{$\bullet$}] $D$ is called {\bf{dominating set}} (or {\it{externally stable}} according to Claude Berge \cite{RefB}) if every vertex $v \in V$ is either an element of $D$ or is adjacent to an element of $D$. Moreover, $D$ is called {\it{minimal dominating set}} if no proper subset $D' \subset D$ is a dominating set.
\item[\small{$\bullet$}] $D$ is called a {\bf{minimum dominating set}} if the cardinality of $D$ is minimum among the cardinalities of any other dominating set.
\item[\small{$\bullet$}] The cardinality of a minimum dominating set $D$ is called the {\bf{domination number}} of graph $G$ and is denoted by $\gamma = \gamma(G)$.
\item[\small{$\bullet$}] Furthermore, $D$ is called {\bf{independent}} set (or {\it{internally stable}} set in the terminology of \cite{RefB}) whenever no two vertices of $D$ are adjacent. Note that an independent set $D$ is also a dominating set if and only if $D$ is a maximal independent set, in which case $D$ is called {\bf{independent dominating set}}.\footnote{The minimum cardinality of an independent dominating set of $G$ is called the {\bf{independent domination number}} of graph $G$ and is denoted by $i = i(G)$. Note that the minimum dominating set of $G$ is not necessarily independent in $G$ and, in general, $\gamma(G) \leq i(G)$.}
\end{itemize}
\end{definition}

The standard reference to graph domination theory is the book of T.W. Haynes, S.T. Hedetniemi and P.J. Slater, \textit{Fundamentals of Domination in Graphs} \cite{RefHHS}, where the proofs of the following basic results can be found.\\

A vertex $v \in D$ is called an {\it{enclave}} of $D$ if $N[v] \subseteq D$ and $v$ is called an {\it{isolate}} of $D$ if  $N(v) \subseteq V \smallsetminus D$. A set is called {\it{enclaveless}} if it does not contain any enclaves. 
Apparently, $D$ is is a dominating set if and only if one of the following holds (cf. \cite{RefHHS}):
\begin{itemize}
\item[\small{$\bullet$}] $N[D] = V$,
\item[\small{$\bullet$}] for every $v \in V \smallsetminus D$, $d(v, D) \leq 1$,
\item[\small{$\bullet$}] $V \smallsetminus D$ is enclaveless.
\end{itemize}

The first existence theorems about dominating sets in graphs were given by O. Ore in his 1962 book, \textit{Theory of Graphs} \cite{RefO}, and they are included in the next theorem and the corollary that follows:

\begin{theorem}
\label{Ore1}
\cite{RefO} A dominating set $D$ is a minimal dominating set if and only if, for each vertex $u \in D$, one of the following two conditions holds:
\begin{itemize}
\item[\small{$\bullet$}] either $u$ is an isolate of $D$, 
\item[\small{$\bullet$}] or there exists a vertex $v \in V \smallsetminus D$ such that $N[v] \cap D = \{u\}$, in which case $v$ is called a {\bf{private neighbor}} of $u$.
\end{itemize}
\end{theorem}

\begin{corollary}
\label{Ore2}
\cite{RefO} Let $G$ be a graph with no isolated vertices. Then:
\begin{itemize}
\item[\small{$\bullet$}] $G$ has a dominating set $D$ and $D$'s complement $V \smallsetminus D$ is also a dominating set.
\item[\small{$\bullet$}] $\gamma(G) \leq \frac{n}{2}$.
\end{itemize}
\end{corollary}

Let us, from now on, assume that $G$ contains no isolated vertices. Furthermore, let us denote by $\mathcal{D} = \mathcal{D}(G)$ the set of all (minimum) dominating sets of graph $G$. Then we claim that we can partition the set of vertices $V$ in  three distinct types of vertices, which are defined as follows:

\begin{definition} Let $v \in V$.
\begin{itemize}
\item[\small{$\bullet$}] Vertex $v$ is said to be {\bf{always dominant}} if, for every $D \in \mathcal{D}$, $v \in D$. We denote by $\frak{A} = \frak{A}(G)$ the set of all always dominant vertices of $G$.
\item[\small{$\bullet$}] Vertex $v$ is said to be {\bf{possibly dominant}} if there exists $D \in \mathcal{D}$ such that $v \in D$. We denote by $\frak{P} = \frak{P}(G)$ the set of all possibly dominant vertices of $G$.
\item[\small{$\bullet$}] Vertex $v$ is said to be {\bf{never dominant}} if, for every $D \in \mathcal{D}$, $v \notin D$. We denote by $\frak{N} = \frak{N}(G)$ the set of all never dominant vertices of $G$.
\end{itemize}
Thus, there exists the following partition of $V$, that we are going to call {\bf{domination partition}}, 
\[
V =  \frak{A} \cup \frak{P} \cup \frak{N}.
\]
Note that at least one of the sets $\frak{A}, \frak{P}, \frak{N}$ must be nonempty.
\end{definition}

Below, we write $G - v$ for the induced subgraph $G(V \smallsetminus \{v\})$. 
Our main result is the following:

\begin{lemma}
Let $\gamma$ be the domination number of $G$ and let $v$ be a vertex. 
\begin{description}
\item[(a)] Vertex $v$ is always a dominant vertex if and only if $\gamma(G - v) > \gamma$, i.e.,
\[
\frak{A} = \{v \in V\!: \gamma(G - v) > \gamma\}.
\]
\item[(b)] If $|\frak{A}| = \gamma$, then 
\begin{eqnarray*}
\frak{N} &=& V \smallsetminus \frak{A},\\ 
\frak{P} &=& \varnothing.
\end{eqnarray*}
In other words, if $|\frak{A}| = \gamma$, then there exists a unique (minimum) dominating set of $G$ (i.e., $|\mathcal{D}(G)| = 1$).
\item[(c)] If $|\frak{A}| := \gamma_1 < \gamma$, then there exists a set $\frak{R} \subset V \smallsetminus \frak{A}$, with $|\frak{R}| = \gamma - \gamma_1$ ($> 0$), such that 
$\frak{A} \cup \frak{R} \in  \mathcal{D}(G)$, and, therefore,
\begin{eqnarray*}
\frak{P} &=& \{\frak{R} \subset V \smallsetminus \frak{A}\!: |\frak{R}| = \gamma - \gamma_1 {\mbox{ and }} \frak{A} \cup \frak{R} \in  \mathcal{D}(G)\},\\ 
\frak{N} &=& V \smallsetminus (\frak{A} \cup \frak{P}).
\end{eqnarray*}
\end{description}
\end{lemma}

\begin{proof} Part (a) follows from the ``uniqueness'' proof of Gunther, Hartnell, Markus and Rall \cite{RefGHM}. Part (b) is plain to see, while (c) is a consequence of Ore's fundamental results (Theorem 1 and Corollary 1).
\end{proof}

%In particular, we may define an index $m = m(G)$ that we call {\bf{index of domination multiplicity}} as $m = 1  - \frac{|\frak{A}|}{\gamma}$.
%\[
%m = 1  - \frac{|\frak{A}|}{\gamma}.
%\]

\begin{definition}
For a graph $G$, the {\bf{index of domination multiplicity}} $m = m(G)$ is defined as
\[
m = 1  - \frac{|\frak{A}|}{\gamma}.
\]
\end{definition}

Clearly, $m$ varies in $[0, 1]$ and it is $m = 0$, for any graph with unique dominating set, while $m = 1$, for a graph without always dominant  vertices (for instance, this is the case with cycle graphs). In this way, when an applied analyst knows the domination number together with the index of domination multiplicity, she possesses comprehensive information in order to be able to infer  the ``complexity'' of the network structure at the global level. Of course, in practical situations, this knowledge works like a new sort of centrality index, which, dissimilarly to the local computation for most centrality indices, now it is based on a combinatorial assessment of all possible structural circumstances enabling the constitution of the domination partition of an empirical network at hand. In this context, the number of always dominant vertices is inversely proportional to the size of the domination index and directly proportional to the graph domination number.

In our numerical computations, after using the SageMath software \cite{sage} to compute the domination number of a graph,  we are computing all the minimum dominating sets of the graph through the following two algorithms that we have implemented in Python.

\begin{algorithm}[H]
    \SetKwInOut{Input}{Input}
    \SetKwInOut{Output}{Output}
    \underline{function AllwaysDom} $(G,\gamma)$\;
    \Input{Graph $G$ and domination number $\gamma$ of $G$.}
    \Output{The set $\frak{A}(G)$ of the always dominant vertices.}
    \For{$\node \in G.\nodes$}
    {\label{forins}
        \If{$\gamma(G \smallsetminus \node) > \gamma$}
    {
        $\frak{A}(G) \leftarrow \node$\;
    }}
\caption{Algorithm for finding the set $\frak{A}(G)$ of all always dominant vertices of graph $G$, when the domination number $\gamma = \gamma(G)$ is known.}

\end{algorithm}
\begin{algorithm}[H]
    \SetKwInOut{Input}{Input}
    \SetKwInOut{Output}{Output}
    \underline{function AllDoms} $(G,\frak{A}(G),\gamma)$\;
    \Input{Graph $G$, the set $\frak{A}$ of all always dominant vertices and the domination number $\gamma$ of $G$.}
    \Output{The collection $\mathcal{D}$ of all minimum dominating sets of $G$.}
    $r \leftarrow \gamma  - |\frak{A}|$\;
    merger $\leftarrow$ the set all subsets of $G \smallsetminus \frak{A}$ of size $r$\;
    \For{$S \in \merger$}
    {\label{forins}
       $S \leftarrow S \cup \frak{A}$\;
    {
        \If{$S \mathrm{ \ is \ Dominating \ Set}$}
        {$\mathcal{D} \leftarrow S$}
    }}
\caption{Algorithm for finding the collection $\mathcal{D} = \mathcal{D}(G)$ of all the minimum dominating sets of graph $G$, when the domination number $\gamma = \gamma(G)$ is known.}
\end{algorithm}

\section{Ego--Centered Decompositions in Graphs}
\label{egod}

\begin{definition}
Let $G = (V, E)$ be a simple undirected graph, $u \in U$ and $U \subset V$. The subgraph $G(U)$ induced by $U$ is called {\bf{ego--centered}} subgraph if 
\[
U = N[u],
\]
in which case vertex $u$ is called the {\bf{ego}} of $G(N[u])$ and all vertices $w \in N(u)$ are called {\bf{alters}} of $G(N[u])$. Moreover, an alter $w \in N(u)$ is called:
\begin{itemize}
\item[\small{$\bullet$}] {\bf{private alter}}  if $N(w) \subset N[u]$ and
\item[\small{$\bullet$}] {\bf{public alter}}  if $N(w) \smallsetminus N[u] \neq \varnothing$.
\end{itemize}
\end{definition}

For $u_1, u_2$ two vertices in $G$, let us consider the two ego--centered subgraphs  $G(N[u_1]), G(N[u_2])$.
\begin{itemize}
\item[\small{$\bullet$}] If $w \in N[u_1] \cap N[u_2]$ ($\neq \varnothing$), then vertex $w$ is called {\bf{shared alter}} by the two ego--centered subgraphs. Note that, as far as $u_1 \neq u_2$, a shared alter is necessarily a public alter.
\end{itemize}
Furthermore, let $w_1 \in N(u_1), w_2 \in N(u_2)$ be two alters such that $w_1 \neq w_2$ and $w_1 \sim w_2$. Note that, whenever $u_1 = u_2$, $w_1,  w_2$ become alters in the same ego--centered subgraph. Then we have the following cases:
\begin{itemize}
\item[\small{$\bullet$}] If $u_1 \neq u_2$ and $u_1 \sim u_2$, then edge $(u_1, u_2)$ is called {\bf{bridge of egos}}.
\item[\small{$\bullet$}] If $w_1$ is a private alter in $G(N[u_1])$ and $w_2$ is a private alter in $G(N[u_2])$, then edge $(w_1, w_2)$ is called {\bf{bridge of private alters}}.
\item[\small{$\bullet$}] If $w_1$ is a public alter in $G(N[u_1])$ and $w_2$ is a public alter in $G(N[u_2])$, then edge $(w_1, w_2)$ is called {\bf{bridge of public alters}}.
\item[\small{$\bullet$}] If $w_1$ is a private alter in $G(N[u_1])$ and $w_2$ is a public alter in $G(N[u_2])$, then edge $(w_1, w_2)$ is called {\bf{bridge of private--to--public alters}}.
\end{itemize}

\begin{definition}
For any minimum dominating set $D \in \mathcal{D}(G)$, the family of ego-centered subgraphs $\{G(N[u])\}_{u \in D}$ is called {\bf{$D$--dominating ego--centered decomposition}} of graph $G$.
Moreover, this decomposition induces the following partitions of the vertex set $V$ and the edge set $E$:
\begin{eqnarray*}
V &=& D \cup V_{\private} \cup V_{\public},\\
E &=& E_{\ego} \cup E_{\private} \cup E_{\public} \cup E_{\private-\public},
\end{eqnarray*}
where the sets of private and (shared) public vertices, $V_{\private} = V_{\private}(D), V_{\public} = V_{\public}(D)$, and the sets of ego, private, public and private--public bridges, $E_{\ego} = E_{\ego}(D), E_{\private} = E_{\private}(D), E_{\public} = E_{\public}(D), E_{\private-\public} = E_{\private-\public}(D)$, are defined as follows:
\begin{eqnarray*}
V_{\private} &=& \{w \in V \smallsetminus D\!: \exists  u \in D {\mbox{ such that }} w \in N(u) \smallsetminus D \\
 & & \,\, {\mbox{ and }}  N(w) \subset N[u]\},\\
V_{\public} &=& \{w \in V \smallsetminus D\!: \exists  u_1, u_2 \in D {\mbox{ such that }} u_1 \neq u_2, \\
 & & \,\,\,\, N(u_1) \cap N(u_2) \neq  \varnothing {\mbox{ and }} w \in N(u_1) \cap N(u_2)\},\\
E_{\ego} &=& \{(u_1, u_2) \in E\!: u_1 \neq u_2 {\mbox{ and }} u_1, u_2 \in D\},\\
E_{\private} &=& \{(w_1, w_2) \in E\!: w_1 \neq w_2 {\mbox{ and }} w_1, w_2 \in V_{\private}\},\\
E_{\public} &=& \{(w_1, w_2) \in E\!: w_1 \neq w_2 {\mbox{ and }} w_1, w_2 \in V_{\public}\},\\
E_{\private-\public} &=& \{(w_1, w_2) \in E\!: w_1 \neq w_2 {\mbox{ and }} w_1 \in V_{\private}, w_2 \in V_{\public}\}.
\end{eqnarray*}
\end{definition}

\begin{proposition}
For any $D$--dominating ego--centered decomposition of graph $G$ with domination number $\gamma$,
\begin{eqnarray*}
|V| - |V_{\private}| - |V_{\public}|  &=& \gamma ,\\
|E| - |E_{\private}| - |E_{\public}|  &=& \sum_{u \in D} \deg(u) - |E_{\ego}| + |E_{\private-\public}|.
%\gamma + |V_{\private}| + |V_{\public}| &=& |V|,\\
%|E_{\private}| + |E_{\public}| + |E_{\private-\public}| &=& |E| + |E_{\ego}| - \sum_{u \in D} \deg(u).
\end{eqnarray*}
\end{proposition}

\begin{proof}
This is a direct consequence of the previous definitions and the degree sum formula (handshaking lemma) for graph $G$.
\end{proof}

\begin{corollary}
Let $\{G(N[u])\}_{u \in D}$ be a $D$--dominating ego--centered decomposition of $G$. Then we have:
\begin{itemize}
\item[\small{$\bullet$}] $D$ is an independent dominating set if and only if $E_{\ego} = \varnothing$.
\item[\small{$\bullet$}] $G$ is the disjoint union of the ego-centered subgraphs  $\{G(N[u])\}_{u \in D}$, i.e., $G = \sum_{u \in D} G(N[u])$, if and only if $V_{\public} = E_{\ego} = E_{\public} = E_{\private-\public} = \varnothing$,  in which case $G$ is disconnected to a forest of ego--centered subgraphs.
\end{itemize}
\end{corollary}

Let us also remark that, for each graph $G$, one may possibly choose different minimum dominating sets $D$ and, thus, different $D$--dominating ego-centered decompositions. However, all these choices depend on the type of the graph $G$ through the value of the index of domination multiplicity $m$ for $G$. Apparently, when $m$ tends to $0$, then  there are fewer choices of minimum dominating sets than when $m$ tends to $1$. Moreover, in case of multiple minimum dominating sets, Proposition 1 suggests that one may select a $D$--dominating ego-centered decomposition according to whether this decomposition maximizes or minimizes the following quantities:
\begin{itemize}
\item[(1)] $\sum_{u \in D} \deg(u)$, i.e., the density of connections between egos and alters;
\item[(2)] $|E_{\ego}|$, i.e., the density of connections among egos;
\item[(3)] $|E_{\private}|$, i.e., the density of connections among private alters;
\item[(4)] $|E_{\public}|$, i.e., the density of connections among public alters;
\item[(5)] $|E_{\private-\public}|$, i.e., the density of connections between private and public alters.
\end{itemize}
One may notice that criterion (1) represents a condition on the total valence of egos and criterion (5) focuses on the separability among the ego--centered subnetworks within a given decomposition. Moreover, the three criteria (2), (3) and (4) are based on the cohesiveness of the set of egos, private alters or public alters, respectively. As a matter of fact, (2), (3) and (4) measure the number of triangles with one of their vertices being an ego and the other two being private alters or public alters, respectively. Similarly, (5) measures the number of 3--paths between pairs of egos. Nevertheless, according to Proposition 1, the five criteria are interdependent in the sense that, for a given network $G$, they cannot be all maximized or minimized simultaneously, since $\sum_{u \in D} \deg(u) + |E_{\private}| + |E_{\public}| - |E_{\ego}| + |E_{\private-\public}| = |E|$.

\section{Applications}
\label{appls}

\subsection{Hage \& Harary's Voyaging Network}
\label{appls1}

The first example of our computation is on Hage and Harary's \cite{RefHH} voyaging network among the 14 western Carolines Islands: Satawal, Ulithi, Woleai, Puluwat, Faraulep, Fais, Pulusuk, Pulap, Elato, Ifaluk, Sorol, Namonuito, Eauripik and Lamotrek. For the corresponding graph $G$, the domination number is equal to 3, there are two always dominant vertices (Elato and Fais) and two possibly dominant vertices (Puluwat and Pulap). The domination partition of this network is plotted in Figure 1.

\begin{figure}[H]
\resizebox{1.4\columnwidth}{!}{%1.4
\hspace*{-5cm}
\includegraphics{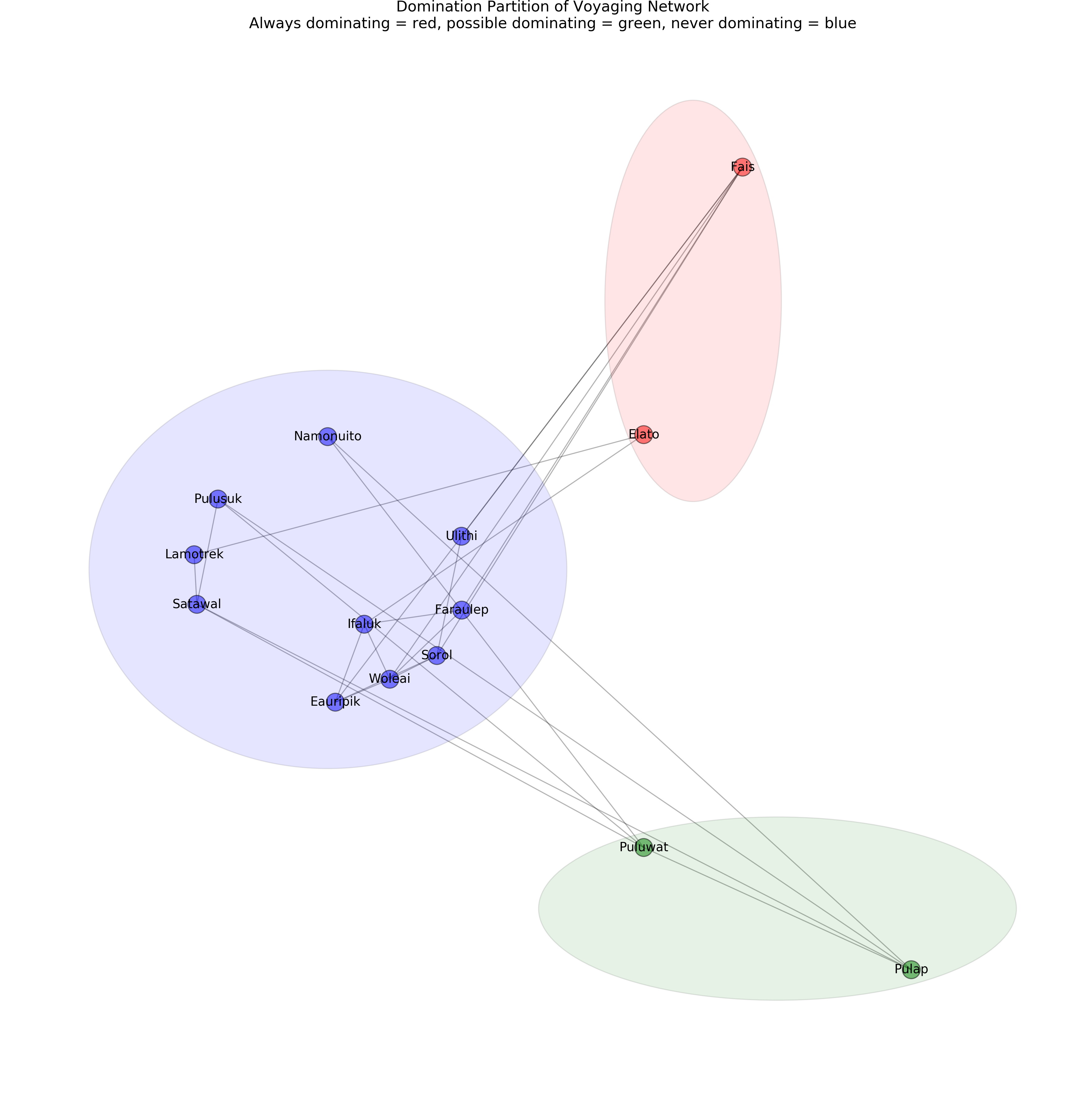}
}
\caption{The Hage \& Harary's Voyaging Network.}
\label{fig:vn}
\end{figure}

In Harary's voyaging network, the index of domination multiplicity is $m = 0.333$ and the properties of the two minimum dominating sets are shown in Table 1. We observe that both minimum dominating sets [Elato, Fais, Puluwat] and [Elato, Fais, Pulap] are independent sets, their egos have sum of degrees equal to 11 and the number of bridges among their private--public alters is 4. In other words, this is an illustration of a relatively simple network (low domination number, low multiplicity number) with two minimum dominating sets that seem to be more or less equally desirable since the rows for them in Table 1 are identical.

\begin{table}[H]
\small 
\caption{The Minimum Dominating Sets (MDSs) of the Hage \& Harary's voyaging network (E = Elato, F = Fais, Pulu = Puluwat, Pul = Pulap).}
\label{table:vn}
\begin{tabular}{lccccccc}
\bf{MDSs}  & $|V_{\private}|$ & $|V_{\public}|$ & $\sum_{v \in S} \deg(v)$ & $|E_{\private}|$ & $|E_{\public}|$ & $|E_{\private-\public}|$ & $|E_{\ego}|$ \\

E, F, Pulu &  5 & 6 & 11 & 3 & 6 & 4 & 0 \\ 

E, F, Pul &  5 & 6 & 11 & 3 & 6 & 4 & 0 \\ 

\end{tabular}
\end{table}

For those who find examples helpful in understanding concepts, let us indicate the egos, private and public alters corresponding to the decompositions induced by the two minimum dominating sets of Harary's voyaging network. 

\begin{figure}[H]
\resizebox{1\columnwidth}{!}{%1.2
%\hspace*{-5cm}
\includegraphics{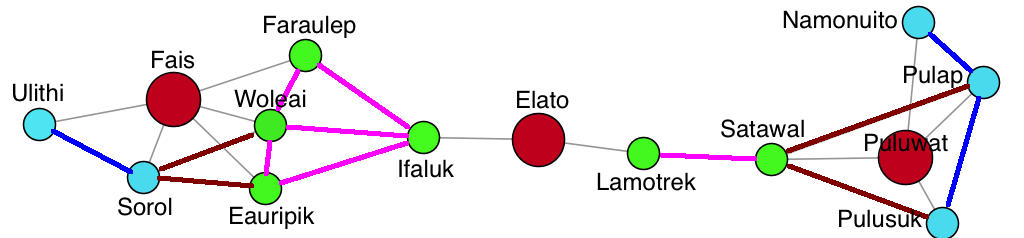}
}
\caption{Egos and alters in the Hage \& Harary's voyaging network.}
\label{fig:kk}
\end{figure} 

For the first minimum dominating set, Elato, Fais and Puluwat are the three egos (colored red in Figure 2). Ego Elato has two alters, Haluk and Lamotrek, both being public alters (colored green). Ego Fais has five alters, two of them, Ulithi and Sorol, being private alters (colored cyan) and the other three, Eauripik, Woleai  and Farauleo, being public alters. Ego Puluwat has four alters, one of them, Satwal, being public alter and three, Pulap, Namonuito and Pulusuk, being private alters. Thus, in total, the minimum dominating set of  Elato, Fais and Puluwat has five private alters (Ulithi, Sorol, Pulap, Namonuito and Pulusuk) and six public alters (Haluk, Lamotrek, Eauripik, Woleai, Farauleo and Satwal).

The second minimum dominating set is composed of three egos, Elato, Fais and Pulap. The alters of egos Elato and Fais have been already identified. Ego Pulap has also four alters, one of them, Satwal, being public alter and three, Puluwat, Namonuito and Pulusuk, being private alters. Thus, in total, the minimum dominating set of  Elato, Fais and Puluwat has five private alters (Ulithi, Sorol, Puluwat, Namonuito and Pulusuk) and six public alters (Haluk, Lamotrek, Eauripik, Woleai, Farauleo and Satwal).

\subsection{Krackhardt's Kite}
\label{appls2}

Next, we examine Krackhardt's Kite network \cite{RefKn} composed of 10 vertices. The corresponding graph $G$ has domination number  equal to 2, but there is a unique minimum dominating set and, thus, there are no possibly dominant vertices. The domination partition of $G$ is plotted in Figure 3. 

\begin{figure}[H]
\resizebox{1.1\columnwidth}{!}{%1.2
\hspace*{-5cm}
\includegraphics{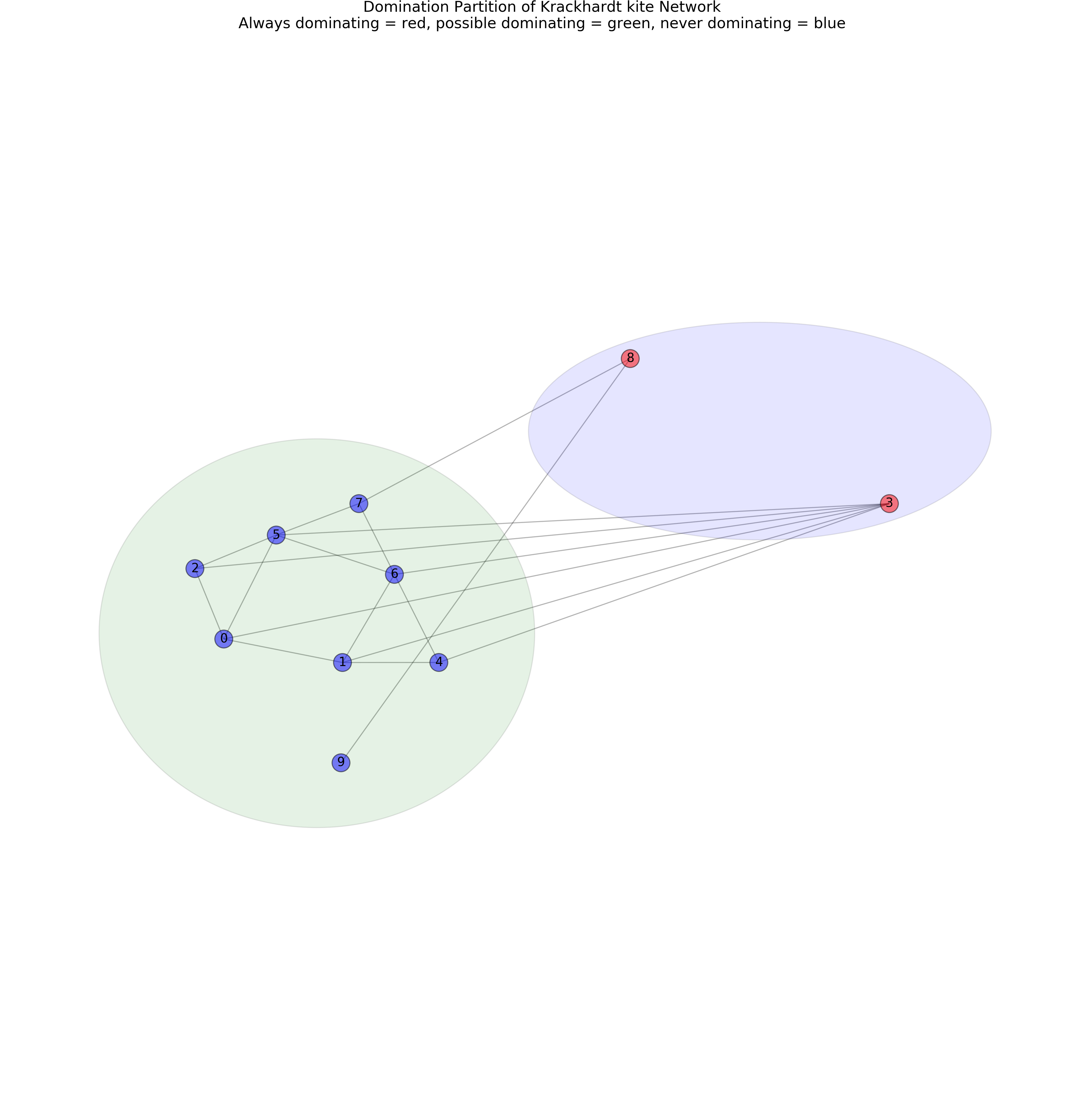}
}
\caption{The Krackhardt kite.}
\label{fig:kk}
\end{figure} 

The index of domination multiplicity of Krackhardt's Kite is $m = 0$ and the properties of the unique minimum dominating set are shown in Table 2: [3, 8] is an independent dominating set, the single ego has sum of degrees equal to 8 and the number of bridges among  private--public alters is 4.

\begin{table}[H]
\small
\caption{The Minimum Dominating Sets (MDSs) of the Krackhardt kite network (with the traditional labelling of vertices: 1 = Andre, 2 = Beverley, 3 = Carol, 4 = Diane, 5 = Ed, 6 = Fernando, 7 = Garth, 8 = Heather, 9 = Ike, 10 = Jane).}
\begin{tabular}{lccccccc}
\bf{MDSs}  & $|V_{\private}|$ & $|V_{\public}|$ & $\sum_{v \in S} \deg(v)$ & $|E_{\private}|$ & $|E_{\public}|$ & $|E_{\private-\public}|$ & $|E_{\ego}|$ \\

[3, 8] &  5 & 3 & 8 & 3 & 3 & 4 & 0 \\ 
    
\end{tabular}
\end{table}

\subsection{Florentine Families}
\label{appls3}

The third example of our computation is the Florentine Families network in the form that Breiger and Pattison have used \cite{RefBP} (extracted from a subset of data on the social relations among Renaissance Florentine families collected by John Padgett). The 15 vertices of this network are the Florentine families: Acciaiuoli, Albizzi, Barbadori, Bischeri, Castellani, Ginori, Guadagni,  Lamberteschi, Medici, Pazzi, Peruzzi, Ridolfi, Salviati, Strozzi and Tornabuoni. The corresponding graph $G$ has domination number equal to 5, there is a single always dominant vertex (Medici) and 9  possibly dominant vertices, shown in the domination partition plot of Figure 4.

\begin{figure}[H]
\resizebox{1.25\columnwidth}{!}{%
\hspace*{-5cm}
\includegraphics{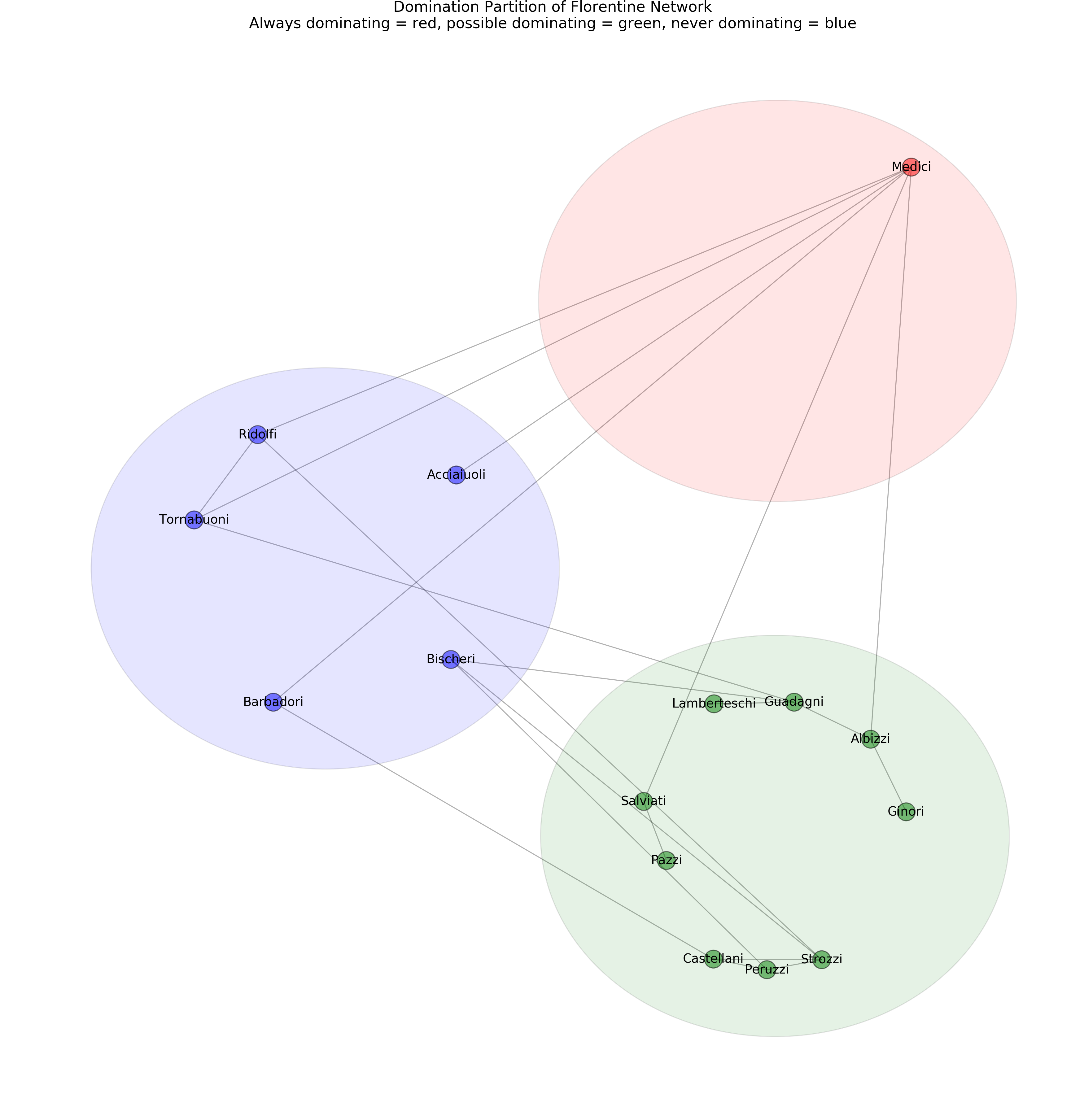}
}
\caption{The Florentine Families network.}
\label{fig:ff}
\end{figure}

Now, the index of domination multiplicity is $m = 0.8$ and the properties of the 20 minimum dominating sets are shown in Table 3. Notice that [Albizzi, Guadagni,  Medici, Salviati, Strozzi] is the dominating set with maximum sum of egos' degrees (equal to 19), there are 5 dominating sets with zero number of bridges among egos (i.e., 5 independent dominating sets) and 12 dominating sets with zero number of bridges among private--public alters.

\begin{table}[H]
\small 
\caption{The Minimum Dominating Sets (MDSs) of the Florentine Families network (0 = Albizzi, 1 = Castellani, 2 = Ginori, 3 = Guadagni, 4 = Lamberteschi, 5 = Medici, 6 = Pazzi, 7 = Peruzzi, 8 = Salviati, 9 = Strozzi).}
\begin{tabular}{lccccccc}
\bf{MDSs}  & $|V_{\private}|$ & $|V_{\public}|$ & $\sum_{v \in S} \deg(v)$ & $|E_{\private}|$ & $|E_{\public}|$ & $|E_{\private-\public}|$ & $|E_{\ego}|$ \\

[1, 2, 3, 5, 6] &  2 & 8 & 15 & 0 & 5 & 0 & 0 \\ 
    
    [0, 4, 5, 6, 7] &  2 & 8 & 14 & 0 & 7 & 0 & 1 \\ 
    
    [2, 4, 5, 8, 9] &  3 & 7 & 14 & 0 & 5 & 2 & 1 \\ 
    
    [0, 3, 5, 8, 9] &  5 & 5 & 19 & 0 & 2 & 2 & 3 \\ 
    
    [2, 3, 5, 6, 7] &  2 & 8 & 15 & 0 & 5 & 0 & 0 \\ 
    
    [2, 3, 5, 7, 8] &  3 & 7 & 16 & 0 & 5 & 0 & 1 \\ 
    
    [0, 3, 5, 6, 7] &  3 & 7 & 17 & 0 & 5 & 0 & 2 \\ 
    
    [2, 4, 5, 6, 7] &  1 & 9 & 12 & 0 & 8 & 0 & 0 \\ 
    
    [2, 3, 5, 8, 9] &  4 & 6 & 17 & 0 & 2 & 2 & 1 \\ 
    
    [0, 4, 5, 7, 8] &  3 & 7 & 15 & 0 & 7 & 0 & 2 \\ 
    
    [2, 3, 5, 6, 9] &  3 & 7 & 16 & 0 & 2 & 2 & 0 \\ 
    
    [0, 4, 5, 8, 9] &  4 & 6 & 16 & 0 & 4 & 2 & 2 \\ 
    
    [0, 1, 3, 5, 6] &  3 & 7 & 17 & 0 & 5 & 0 & 2 \\ 
    
    [2, 4, 5, 6, 9] &  2 & 8 & 13 & 0 & 5 & 2 & 0 \\ 
    
    [0, 3, 5, 7, 8] &  4 & 6 & 18 & 0 & 5 & 0 & 3 \\ 
    
    [2, 4, 5, 7, 8] &  2 & 8 & 13 & 0 & 8 & 0 & 1 \\ 
    
    [0, 4, 5, 6, 9] &  3 & 7 & 15 & 0 & 4 & 2 & 1 \\ 
    
    [0, 1, 3, 5, 8] &  4 & 6 & 18 & 0 & 5 & 0 & 3 \\ 
    
    [0, 3, 5, 6, 9] &  4 & 6 & 18 & 0 & 2 & 2 & 2 \\ 
    
    [1, 2, 3, 5, 8] &  3 & 7 & 16 & 0 & 5 & 0 & 1 \\
    
    \end{tabular}
\end{table}

\subsection{Karate Club}
\label{appls4}

The fourth example is Zachary's Karate Club network \cite{RefZa} composed of 34 vertices. The corresponding graph $G$ has domination number  equal to 4, there are two always dominant vertices (0 and 33) and 6 possibly dominant vertices (5, 6, 16, 24, 25, 31), as they are shown in the domination partition plot of Figure 5.  

\begin{figure}[H]
\resizebox{1.3\columnwidth}{!}{%
\hspace*{-5cm}
\includegraphics{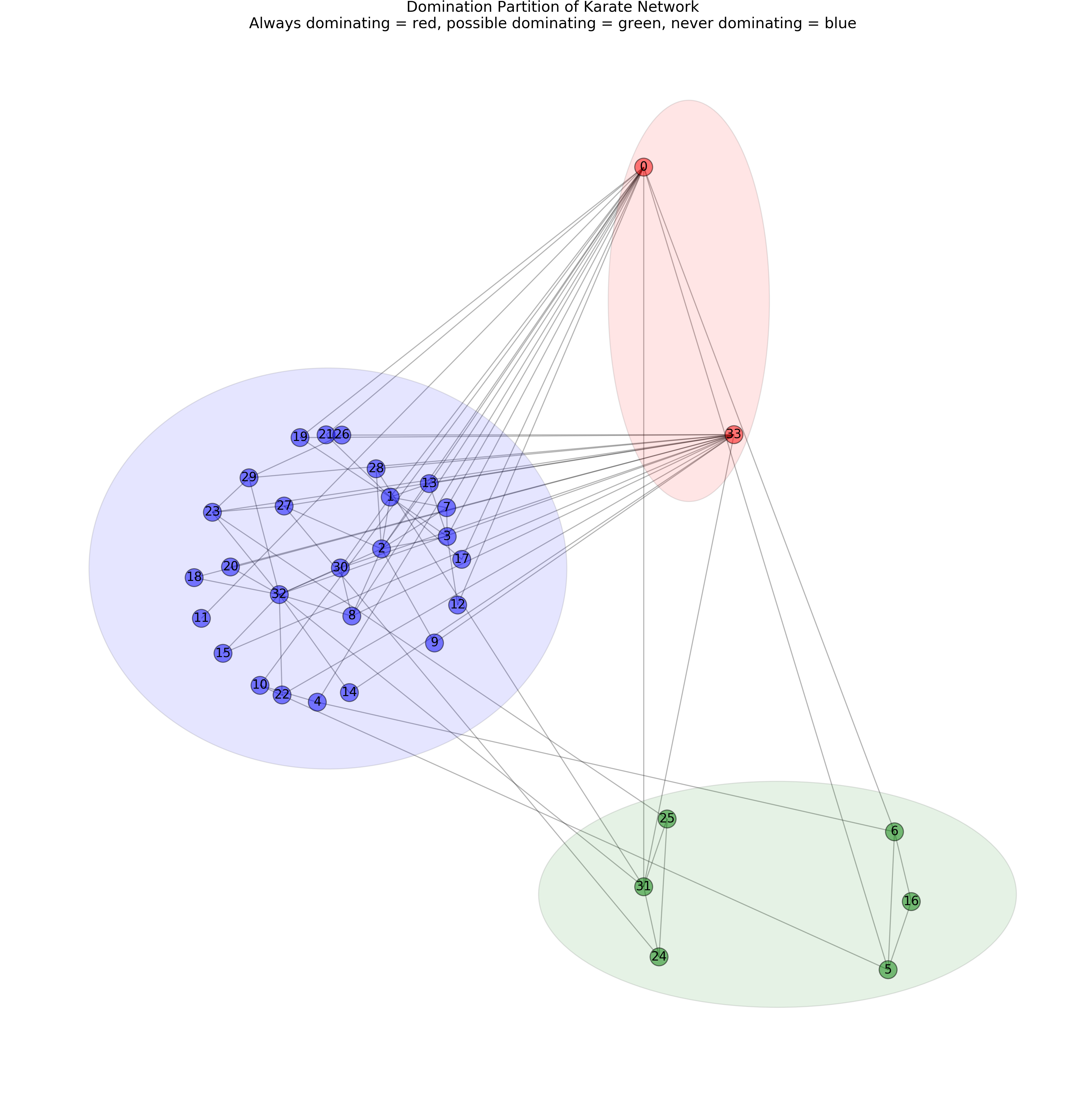}
}
\caption{The Karate Club network.}
\label{fig:kc}
\end{figure} 

In the Karate Club network, the index of domination multiplicity is $m = 0.5$ and the properties of the 9 minimum dominating sets are shown in Table 4. There are 2 dominating sets with maximum sum of egos' degrees (equal to 43), 2 dominating sets with zero number of bridges among egos (i.e., 2 independent dominating sets) and 3 dominating sets with minimum number of bridges among private--public alters (equal to 16).

\begin{table}[H]
\small
\caption{The Minimum Dominating Sets (MDSs) of the Karate Club Network.}
\begin{tabular}{lccccccc}
\bf{MDSs}  & $|V_{\private}|$ & $|V_{\public}|$ & $\sum_{v \in S} \deg(v)$ & $|E_{\private}|$ & $|E_{\public}|$ & $|E_{\private-\public}|$ & $|E_{\ego}|$ \\

[0, 5, 31, 33] &  15 & 15 & 43 & 3 & 18 & 17 & 3 \\ 
    
    [0, 16, 31, 33] &  15 & 15 & 41 & 4 & 19 & 16 & 2 \\ 
    
    [0, 5, 24, 33] &  15 & 15 & 40 & 3 & 19 & 17 & 1 \\ 
    
    [0, 5, 25, 33] &  15 & 15 & 40 & 3 & 19 & 17 & 1 \\ 
    
    [0, 6, 31, 33] &  15 & 15 & 43 & 3 & 18 & 17 & 3 \\ 
    
    [0, 6, 25, 33] &  15 & 15 & 40 & 3 & 19 & 17 & 1 \\ 
    
    [0, 16, 24, 33] &  15 & 15 & 38 & 4 & 20 & 16 & 0 \\ 
    
    [0, 6, 24, 33] &  15 & 15 & 40 & 3 & 19 & 17 & 1 \\ 
    
    [0, 16, 25, 33] &  15 & 15 & 38 & 4 & 20 & 16 & 0 \\ 
    
    \end{tabular}
\end{table}

\subsection{Southern Women}
\label{appls5}

The next example is the Southern Women network \cite{RefDavis}, which is a two-mode network composed of 18 southern women and 14 informal social events. The corresponding bipartite graph $G$ has domination number  equal to 5, there are no always dominant vertices, but there are 18 possibly dominant vertices (8 women and 10 events), as they are shown in the domination partition plot of Figure 6.  

\begin{figure}[H]
\resizebox{1.25\columnwidth}{!}{%
\hspace*{-5cm}
\includegraphics{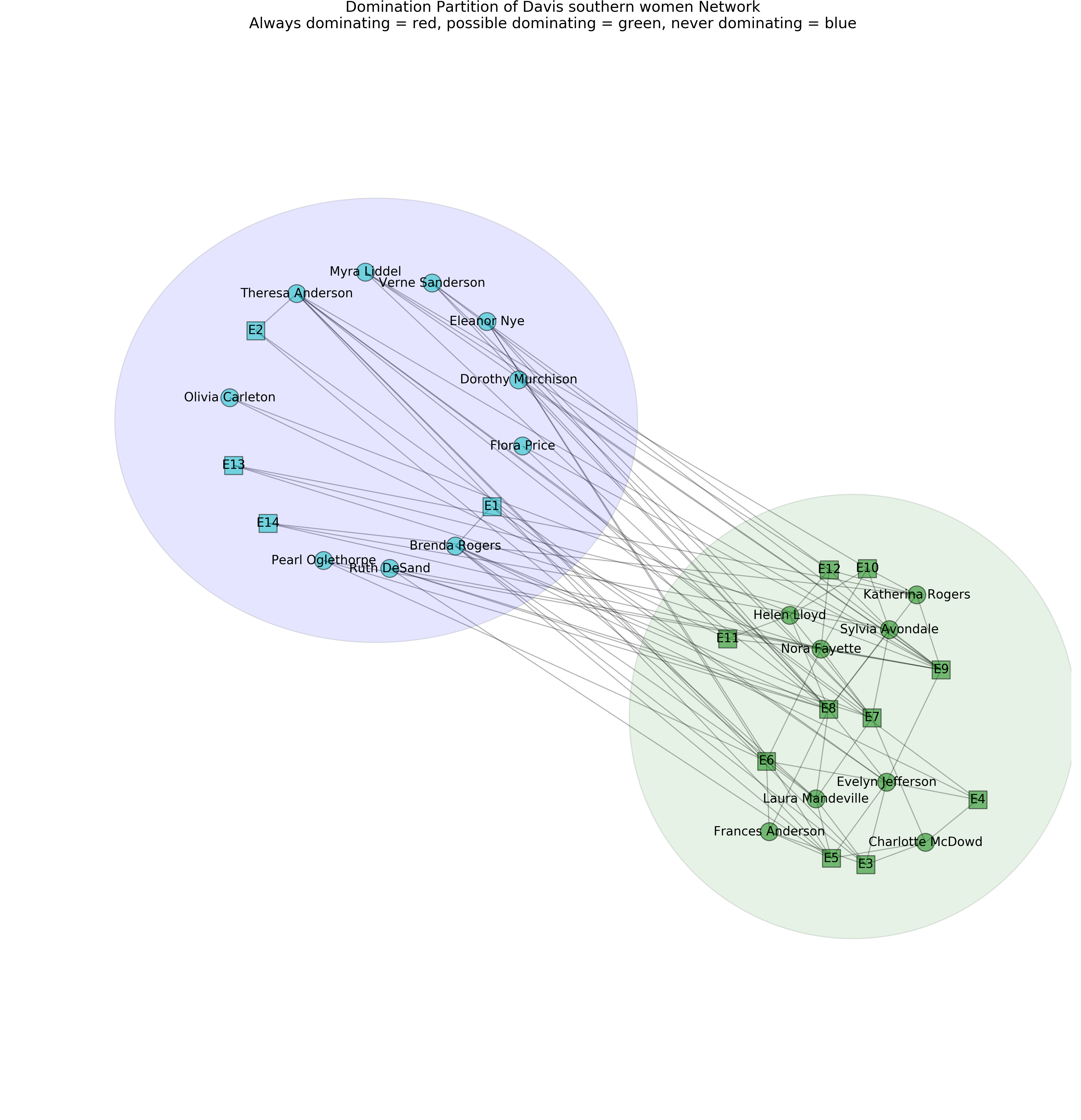}
}
\caption{The Southern Women network.}
\label{fig:lm}
\end{figure}

The index of domination multiplicity of the network of Southern Women is $m = 1$ (since $|\frak{A}| = 0$) and the properties of the 36 minimum dominating sets are shown in Table 5. There is a single dominating set with maximum sum of egos' degrees (equal to 52), all dominating sets have zero number of bridges among egos (because the graph is bipartite) and there are 12 dominating sets with minimum number of bridges among private--public alters (equal to 2).

\begin{table}[H]
\tiny
\caption{The Minimum Dominating Sets (MDSs) of the Southern Women Network.}
\begin{tabular}{lccccccc}
\bf{MDSs}  & $|V_{\private}|$ & $|V_{\public}|$ & $\sum_{v \in S} \deg(v)$ & $|E_{\private}|$ & $|E_{\public}|$ & $|E_{\private-\public}|$ & $|E_{\ego}|$ \\

[0, 1, 6, E8, E11] &  0 & 27 & 38 & 0 & 53 & 0 & 2 \\ 
    
    [1, 6, E3, E7, E9] &  0 & 27 & 44 & 0 & 49 & 0 & 4 \\ 
    
    [1, 7, E7, E8, E11] &  0 & 27 & 43 & 0 & 49 & 0 & 3 \\ 
    
    [1, 6, E7, E8, E11] &  0 & 27 & 44 & 0 & 48 & 0 & 3 \\ 
    
    [5, 7, E4, E8, E11] &  0 & 27 & 36 & 0 & 55 & 0 & 2 \\

    [0, 5, 6, E8, E11] &  0 & 27 & 37 & 0 & 54 & 0 & 2 \\ 
    
    [1, 7, E4, E8, E11] &  0 & 27 & 37 & 0 & 55 & 0 & 3 \\ 
    
    [0, 5, 6, E8, E9] &  0 & 27 & 45 & 0 & 46 & 0 & 2 \\ 
    
    [1, 7, E5, E9, E11] &  0 & 27 & 39 & 0 & 53 & 0 & 3 \\ 
    
    [1, 7, E3, E8, E11] &  0 & 27 & 39 & 0 & 53 & 0 & 3 \\ 
    
    [0,  1, 4, E8, E11] &  0 & 27 & 36 & 0 & 55 & 0 & 2 \\ 
    
    [1, 6, E5, E8, E11] &  0 & 27 & 42 & 0 & 50 & 0 & 3 \\ 
    
    [0, 1, 7, E8, E11] &  0 & 27 & 37 & 0 & 54 & 0 & 2 \\ 
    
    [1, 6, E7, E8, E9] &  0 & 27 & 52 & 0 & 41 & 0 & 4 \\ 
    
    [1, 6, E3, E8, E11] &  0 & 27 & 40 & 0 & 52 & 0 & 3 \\ 
    
    [1, 6, E4, E8, E11] &  0 & 27 & 38 & 0 & 54 & 0 & 3 \\ 
    
    [1, 6, E3, E8, 'E9] &  0 & 27 & 48 & 0 & 45 & 0 & 4 \\ 
    
    [1, 2, 6, E7, E9] &  0 & 27 & 42 & 0 & 50 & 0 & 3 \\ 
    
    [1, 6, E5, E9, E10] &  0 & 27 & 41 & 0 & 52 & 0 & 4 \\ 
    
    [1, 3, 4, E5, E9] &  0 & 27 & 39 & 0 & 53 & 0 & 3 \\ 
    
    [1, 4, E7, E8, E11] &  0 & 27 & 42 & 0 & 49 & 0 & 2 \\ 
    
    [1, 6, E5, E9, E12] &  0 & 27 & 42 & 0 & 51 & 0 & 4 \\ 
    
    [0, 1, 6, E8, E9] &  0 & 27 & 46 & 0 & 46 & 0 & 3 \\ 
    
    [5, 6, E4, E8, E11] &  0 & 27 & 37 & 0 & 54 & 0 & 2 \\ 
    
    [1, 7, E5, E8, E11] &  0 & 27 & 41 & 0 & 51 & 0 & 3 \\ 
    
    [4, 5, E4, E8, E11] &  0 & 27 & 35 & 0 & 56 & 0 & 2 \\ 
    
    [1, 6, E5, E9, E11] &  0 & 27 & 40 & 0 & 53 & 0 & 4 \\ 
    
    [5, 6, E4, E8, E9] &  0 & 27 & 45 & 0 & 46 & 0 & 2 \\ 
    
    [1, 6, E6, E7, E9] &  0 & 27 & 46 & 0 & 48 & 0 & 5 \\ 
    
    [1, 3, 7, E5, E9] &  0 & 27 & 40 & 0 & 52 & 0 & 3 \\ 
    
    [0, 5, 7, E8, E11] &  0 & 27 & 36 & 0 & 55 & 0 & 2 \\ 
    
    [1, 3, 6, E5, E9] &  0 & 27 & 41 & 0 & 51 & 0 & 3 \\ 
    
    [1, 6, E5, E7, E9] &  0 & 27 & 46 & 0 & 47 & 0 & 4 \\ 
    
    [1, 6, E4, E8, E9] &  0 & 27 & 46 & 0 & 47 & 0 & 4 \\ 
    
    [1, 6, E5, E8, E9] &  0 & 27 & 50 & 0 & 43 & 0 & 4 \\ 
    
    [0, 4, 5, E8, E11] &  0 & 27 & 35 & 0 & 56 & 0 & 2 \\
    
    \end{tabular}
\end{table}

\subsection{Les Miserables}
\label{appls6}

The last example is the network  of Les Miserables \cite{RefKLM}, which is the network of co-occurrences of 77 characters in Victor Hugo's novel ``{\it{Les Misérables}}.'' The corresponding graph $G$ has domination number  equal to 10, there are 7 always dominant vertices and 6 possibly dominant vertices, as they are shown in the domination partition plot of Figure 7.

\begin{figure}[H]
\resizebox{1.3\columnwidth}{!}{%
\hspace*{-5cm}
\includegraphics{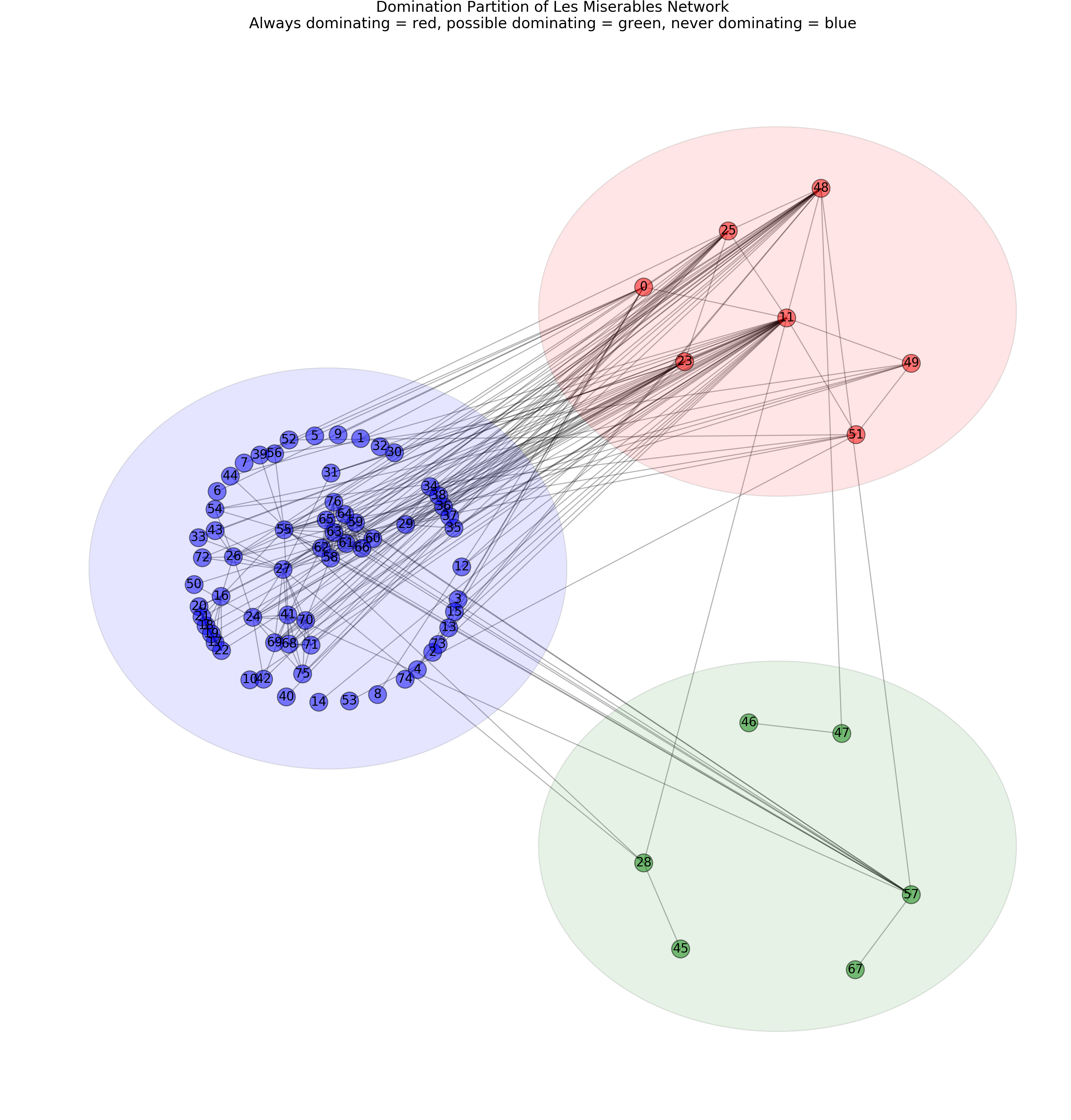}
}
\caption{The network of Les Miserables.}
\label{fig:lm}
\end{figure} 

In the network of Les Miserables, the index of domination multiplicity is $m = 0.3$ and the properties of the 8 minimum dominating sets are shown in Table 6. There is a single dominating set with maximum sum of egos' degrees (equal to 130), a single dominating set with minimum number of bridges among egos (equal to 9) and 2 dominating sets with minimum number of bridges among private--public alters (equal to 39).

\begin{table}[H]
\tiny 
\caption{The Minimum Dominating Sets (MDSs) of the network  of Les Miserables.}
\begin{tabular}{lccccccc}
\bf{MDSs}  & $|V_{\private}|$ & $|V_{\public}|$ & $\sum_{v \in S} \deg(v)$ & $|E_{\private}|$ & $|E_{\public}|$ & $|E_{\private-\public}|$ & $|E_{\ego}|$ \\

\tiny{[0, 11, 23, 25, 45, 46, 48, 49, 51, 67]} &  41 & 26 & 116 & 44 & 53 & 50 & 9 \\ 
    
    \tiny{[0, 11, 23, 25, 28, 46, 48, 49, 51, 67]} &  41 & 26 & 119 & 44 & 52 & 49 & 10 \\ 
    
    \tiny{[0, 11, 23, 25, 28, 47, 48, 49, 51, 57]} &  39 & 28 & 130 & 28 & 69 & 39 & 12 \\ 
    
    \tiny{[0, 11, 23, 25, 28, 46, 48, 49, 51, 57]} &  38 & 29 & 129 & 28 & 69 & 39 & 11 \\ 
    
    \tiny{[0, 11, 23, 25, 28, 47, 48, 49, 51, 67]} &  42 & 25 & 120 & 44 & 52 & 49 & 11 \\ 
    
    \tiny{[0, 11, 23, 25, 45, 47, 48, 49, 51, 57]} &  39 & 28 & 127 & 28 & 70 & 40 & 11 \\ 
    
    \tiny{[0, 11, 23, 25, 45, 47, 48, 49, 51, 67]} &  42 & 25 & 117 & 44 & 53 & 50 & 10 \\ 
    
    \tiny{[0, 11, 23, 25, 45, 46, 48, 49, 51, 57]} &  38 & 29 & 126 & 28 & 70 & 40 & 10 \\ 
    
    \end{tabular}
\end{table}

\section{Conclusion}
\label{concl}

The problem of decomposing a whole network into a minimal number of ego--centered subnetworks has been tackled here. In the graph--theoretic approach that we followed in order to solve this problem, the network egos were selected as the members of a minimum dominating set of the graph (assumed to be undirected here). However, since in general the solution to the graph domination problem may not be unique, we have developed an algorithm that allowed us to compute all minimum dominating sets of a graph, given (knowing) the graph domination number. Our algorithm was based on the partition of the set of vertices into three subsets, the always dominant vertices, the possible dominant vertices and the never dominant vertices. In this way, we managed to associate a dominating ego--centered decomposition to each one of the minimum dominating sets. Moreover, we introduced a number of structural measures through which we achieved comparisons and assessments of any such dominating ego--centered decompositions. In order to illustrate our methodology, we applied it to six empirical networks of various degrees of complexity in terms of graph order, size and multiplicity of minimum dominating sets. 

In a subsequent work, we intend to generalize our approach to the case of directed graphs, for which the concepts of dominating sets have been already elaborated in the graph--theoretic literature. However, we need to do a number of modifications in our algorithm that will enable us to compute all the minimum dominating sets in a directed graph and assign to each one of them a corresponding graph decomposition into ego--centered subgraphs. Subsequently, after extending our methodology, we expect to be able to investigate the domination structure and ego-centered decompositions of such empirical networks as bibliometric citation networks or other cases of directed graphs extracted from social media mining (for instance, networks of RTs from Twitter data).

\end{document}